%% file: xorsat.tex
\documentclass[journal]{IEEEtran}
\newcommand{\SortNoop}[1]{}
\usepackage{amssymb,stmaryrd,amsmath,amsfonts,rotating,mathrsfs}

\usepackage[noadjust]{cite}
\usepackage{epic}
\RequirePackage{bbm}
\input{./definition}
\allowdisplaybreaks

\begin{document}
\title{The Space of Solutions of Coupled XORSAT Formulae}
\author{ S. Hamed Hassani, Nicolas Macris and Rudiger Urbanke \thanks{The authors are with School of Computer
 \& Communication Sciences, EPFL, Switzerland.
}
}

\maketitle

\begin{abstract}
The XOR-satisfiability (XORSAT) problem deals with a system  of $n$
Boolean variables and  $m$  clauses.  Each clause is a linear Boolean
equation (XOR) of a subset of the variables.  A $K$-clause is a
clause involving $K$ distinct variables. In the random $K$-XORSAT
problem a formula is created by choosing $m$ $K$-clauses uniformly
at random from the set of all possible clauses on $n$ variables.
The set of solutions of a random
formula exhibits various geometrical transitions as the ratio $\frac{m}{n}$ varies.

We consider a {\em coupled} $K$-XORSAT ensemble, consisting of a
chain of random XORSAT models that are spatially coupled across a
finite window along the chain direction.  We observe that the
threshold saturation phenomenon takes place for this ensemble and
we characterize various properties of the space of solutions of
such coupled formulae.  \end{abstract}

\section{Introduction}
Spatial coupling is a technique that starts with a graphical model
and a ``hard" computational task (e.g., decoding or more generally
inference) and creates from this a new graphical model for the same
task that has ``locally" the same structure but is computationally
``easy".  Kudekar, Richardson and Urbanke \cite{KRU1,KRU2} made the basic observation (in the context of coding theory)
that  on spatially-coupled graphs, low-complexity (message passing) algorithms suffice to achieve
optimal performance. Despite its very recent introduction, spatial
coupling has already had significant impact on coding, communications,
and compressive sensing 
(see for example \cite{henry2}-\cite{DJM})
and has lead to new insights in computer science and statistical
physics (see \cite{csp}).

We consider the effect of spatial coupling on random XORSAT formulae.
The XORSAT problem is the simplest instance among the class of
constraint satisfaction problems (CSP).  CSPs arise in many branches
of science, e.g., in statistical physics (spin  glasses),
information theory (LDPC codes), and in combinatorial optimization
(satisfiability, coloring).  These CSPs are  believed to
share a number of common structural properties, but some models are inherently
more difficult to investigate than others. It is therefore natural
to start with relatively ``simple'' CSPs if one wants to learn more
about the general behavior of this class of models.

It is relatively simple to capture the same basic properties in the
XORSAT problem due to its direct connection with linear algebra.
Among such properties, an important one is the geometry of the space of solutions, which as was already understood
a decade ago displays very interesting phase transitions \cite{DM,MRZ}. Recently in \cite{mon,ach}, a fairly complete
characterization of this geometry has been provided as a function
of the ratio of number of clauses to  number of variables. In
particular,  it is shown that for some range of values of this
parameter, the space of solutions breaks into many disconnected
``clusters". It is widely believed that such a cluster structure is closely connected to the failure 
 of standard message passing
algorithms to find solutions (e.g., the belief propagation algorithm).
In other words, it is believed that there is a strong connection
between the ``hardness'' of the problem and the geometry of the
solution space.  Therefore we call this regime the {\em hard-SAT}
regime.

Consider now what happens when we spatially couple such formulae.
As we will show in the following, a remarkable phenomenon called
threshold saturation takes place: the belief propagation algorithm
succeeds in solving the problem in the hard-SAT regime of the original (non-coupled) model.  This
immediately raises the question how the space of solutions changes
under spatial coupling. In other words,
what happens to the clusters?  A naive guess is that these clusters
become connected. As we will see, the answer is--yes!

Our main objective is to provide an explanatory picture of how the
geometry of the solution space is altered under spatial coupling.
This picture can be helpful in further understanding the mechanism
of spatial coupling, as well as in gaining some intuition about the
solution space of other coupled CSPs, or in designing efficient
algorithms for solving them \cite{csp}.

The outline of this paper is as follows. In Section~\ref{basic} we
introduce in detail the XORSAT problem and random $K$-XORSAT
ensembles.  We also explain in brief the related results on the
geometry of the solution space of these random formulae. In
Section~\ref{cou} we introduce the coupled $K$-XORSAT ensemble.
Using the results of \cite{henry} and \cite{csp} we then prove the
threshold saturation phenomenon for this ensemble. Finally, we
discuss the geometry of the space of solutions of this ensemble by
a direct use of the techniques in \cite{mon}.

\subsection{The $K$-XORSAT Ensemble: Basic Setting} \label{basic}

An XORSAT formula consists of $n$ Boolean variables $x_i \in \{0,1\}$, $i \in \{1,\cdots,n\}$, and a set of $m$ exclusive OR (XOR) constraints $c \in \{1,\cdots,m\}$. Each constraint, $c$, called from now on a {\em clause}, is a linear equation  consisting of the XOR  of some variables being equal to a Boolean value $b_c \in \{0,1\}$.   The number of variables involved in a clause is called the length of the clause.  Further, a clause of length $K$ is typically called a $K$-clause. Furthermore,  A $K$-XORSAT formula is a formula consisting only of $K$-clauses. In matrix form, a $K$-XORSAT formula can be represented as linear system 
\begin{equation}
\mathbb{H} \underline{x}=\underline{b}. 
\end{equation}
Here, the matrix $\mathbb{H}$ is an $m \times n$ matrix with entries $H_{c,i} \in\{\0,1\}$,  and $H_{c,i}$ is equal to $1$ if and only clause $c$ contains the variable $x_i$.   The vector $\underline{x}$ is an $m$ component vector representing the variables and the vector $\underline{b}$ is also an $m$ component vector representing the clause values $b_c$.
  
It is convenient to represent a XORSAT formula via a bipartite graph $G=(V \cup C,E)$, where we denote the set of variable nodes by $V$ and the set of clause nodes by $C$. We thus have $|V|=n$ and $|C|=m$. There is an edge between a clause $c \in V$ and a variable $i \in V$ if and only if  $c$ contains  $x_i$.   The set of edges of $G$ is denoted by $E$. 
%

Let us now explain the ensemble of random $K$-XORSAT formulae. 
Let $m=\lfloor \alpha  n \rfloor$, where  $\alpha$ is a positive real number and is called {\em the clause density}.  To choose an instance from the $K$-XORSAT ensemble, we proceed as follows. There are $m$ clauses of length $K$ and $n$ variables. Each clause picks uniformly at random a subset of length $K$ of the variables  and flips a fair coin to decide the value of $b_c$.  All the above steps are taken independent of each other. 
In other words, the random $K$-XORSAT  ensemble is defined by taking $\underline{b}$ uniformly at random in $\{0,1\}^m$ and $\mathbb{H}$ uniformly at random from the set of all the $m \times n $ matrices with entries in $\{\0,1\}$ that have exactly $K$ ones per row.  

One objective of the XORSAT problem is to specify whether a given formula has a solution or not. Standard linear algebraic methods allow us to accomplish this task with complexity $O(n^3)$. Here, we discus a linear complexity algorithm for solving XORSAT formulae called the peeling algorithm. In our case, this algorithm is known to be equivalent to the \emph{belief propagation}(BP) algorithm.

\subsection{The Peeling Algorithm}\label{peel}
We begin by a brief explanation of the algorithm. Let $G$ be an XORSAT formula. As mentioned previously, we can think of $G$ as a bipartite graph.
The algorithm starts with $G$ and in each step shortens $G$ until 
we either reach the empty graph or we can not make any further shortening. 
Assume now that there exists a variable $i$ in $G$ with degree $0$ or $1$.
In the former case, the value of the variable can be chosen freely. Also, in the latter case, assuming $c$ is the check node connected to $i$, it is easy to see that the 
value of $x_i$ can be determined after the values of the other variables connected to $c$ are specified. 
Hence, without loss of generality, we can remove $i$ and its neighboring clause (if any) from $G$ and search for a
solution for the graph $G\setminus i$.
In other words, finding a solution for $G$ is equivalent to finding  a solution for  $G\setminus i$.   
As a result, we can peel the variable $i$ from $G$ and do the same procedure on $G\setminus i$. 
We continue this process until the residual graph is empty or  it has no more variables with degree at most $1$. 
The final graph that we reach to by the peeling procedure is called the $2$-core or the maximal stopping set of $G$.
We recall that a stopping set of $G$ is a subgraph of $G$ containing a set of clauses and 
a set of variables where each clause  has degree $K$ and all the variables have degree at least $2$. 
The $2$-core is a stopping set of maximum size.
The peeling algorithm determines the $2$-core of a graph $G$.  
If the $2$-core is empty then the algorithm succeeds and it is easy to see that the solution can be explicitly found by backtracking.

The peeling algorithm has an equivalent message passing (MP) formulation. It can be shown that the message passing rules for the peeling algorithm are
also equivalent to the BP update rules. 
Further, if the formula $G$ comes from the $K$-XORSAT ensemble, then one can analyze the behavior of the peeling algorithm in a probabilistic framework called {\em density evolution} (DE). The DE equations can be cast into a simple scalar recursion \cite{mon-mez}
\begin{equation} \label{pl_rec}
x^{t+1}= 1-\exp \{- \alpha K (x^t) ^{K-1}\}, 
\end{equation}
with $x^0=1$. Here, $x_z^t$ is related to the fraction of edges present in the remaining graph at time $t$.
For the peeling algorithm to succeed, the value of $x^{t}$ should tend to $0$ as $t$ increases. This is possible if and only if the equation
\begin{equation} \label{pl_fp}
x= 1-\exp \{- \alpha K x ^{K-1}\} ,
\end{equation}
has a unique solution which is the trivial fixed point $x=0$. 
The net result is that the  peeling algorithm
 succeeds with high probability (w.h.p)  for $\alpha<\alpha_{d}(K)$ defined as
\begin{multline*}
\alpha_{d}(K) =\\ 
\sup \{ \alpha \geq 0 \,\,\,{\rm{ s.t. }} \,\,\,\forall x \in (0,1]: \,\, x >\alpha (1- \exp(-\frac{ K}{2} x))^{K-1}   \}.
\end{multline*}

For $\alpha>\alpha_d(K)$ the peeling algorithm is w.h.p stuck in the $2$-core of the graph. It can be shown \cite{mon} that the $2$-core consists  w.h.p of $n V(\alpha,K)(1+o(1))$ variables and 
$n C(\alpha,K)(1+o(1))$ clauses, where $V(\alpha,K)$ and $C(\alpha,K)$ are given as follows.
Let $x$ be the largest solution of \eqref{pl_fp} and $\hat{x}=x^{K-1}$, we have  
\begin{align}
&V(\alpha,K)= 1- (1 + K \alpha \hat{x}) \exp(-\alpha K \hat{x}),\\
 & C(\alpha,K)=\alpha \hat{x} (1- \exp(-\alpha K \hat{x})).
\end{align}  

\subsection{Phase Transitions and the Space of Solutions} \label{un_sol}

For a random $K$-XORSAT formula with $\alpha< \alpha_{d}(K)$, the peeling algorithm succeeds w.h.p and hence the formula has a solution. 
What happens for $\alpha> \alpha_{d} (K)$? It is easy to see that for $\alpha>1$ the formula has w.h.p no solution. 
In fact, there exists a critical density $\alpha_s(K)$ such that when the clause density crosses $\alpha_s(K)$, the $K$-XORSAT ensemble undergoes a phase transition from almost certain solvability to almost certain unsatisfiability.  The value $\alpha_s(K)$ is called the SAT/UNSAT threshold and  is given as
\begin{equation} \label{alpha_s}
\alpha_{s}(K) =
\sup \{ \alpha \geq 0 \,\,\,{\rm{ s.t. }} \,\, V(K,\alpha) >C(K,\alpha) \}.
\end{equation}
The value of $\alpha_d$ separates two phases. For $\alpha<\alpha_d$ the graph has no $2$-core whereas for $\alpha \in (\alpha_d, \alpha_s)$ the graph has a large $2$-core and no algorithm is known to find a solution in  linear time. These two phases differ also in the structure of their solution space as we explain now.  We assume without loss of generality that the vector $\underline{b}$ is the all-zero vector. Note here that a 
non-zero $\underline{b}$ affects the solution space of
 the homogeneous system only by a  shift and hence does not alter its structure.       

The solutions of a formula  are members of the Hamming cube $\{0,1\}^n$. For $x,y \in \{0,1\}^n$ we let $d(x,y)$ denote their Hamming distance. For $\alpha < \alpha_d$, there exists a constant $B<\infty$ such that that w.h.p the following holds \cite{mon}. Let $d= (\log n)^B$.  Consider two solutions $\underline{x},\underline{x}'$. Then, there exists a sequence of solutions $\underline{x}=\underline{x}_0, \underline{x}_1, \cdots, \underline{x}_r=\underline{x}'$ such that $d(x_i, x_{i+1}) \leq d$.  Thus, for $\alpha < \alpha_d$, the space of solutions can be imagined as a big cluster in which one can walk from one solution to another by  a numbers steps that are of size at most $d$ (sub-linear in $n$).  
For $\alpha \in (\alpha_d, \alpha_s)$ the space of solutions shatters into an exponential number of clusters. Each cluster corresponds to a solution of the $2$-core in the following sense. Given an assignment $\underline{x}$, we denote by $\pi (\underline{x})$ its projection onto the core. In other words, $\pi (\underline{x})$ is the vector of those entries   in $\underline{x}$ that corresponds to vertices in the core. Now, for a  solution of the core, $x_{\rm core}$, we define the cluster associated to $x_{\rm core}$ as the set of solutions to the whole formula such that $\pi (\underline{x})=x_{\rm core}$.  Hence, for each solution of the core, there exists one cluster in the space of solutions of the formula. It can be shown that each two solutions of the core differ in $\Theta(n)$ positions \cite{mon-mez}. Thus, any two solutions belonging to two different clusters  also differ in $\Theta(n)$ positions.  However, each cluster by itself has a connected structure in the sense that for any two solutions
$\underline{x},\underline{x'}$ belonging to the cluster, there exists a sequence of solutions
 $\underline{x}=\underline{x}_0, \underline{x}_1, \cdots, \underline{x}_r=\underline{x'}$ inside that cluster such that $d(x_i, x_{i+1}) \leq d$.   
 Figure~\ref{fig:space} shows a symbolic picture of the clustering of solutions in the two phases.  
  \begin{figure}[ht!]
 \begin{center} 
  \includegraphics[width=8cm]{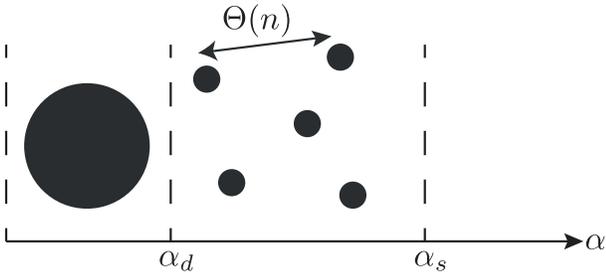}
 \end{center}
 \caption{A symbolic picture of the space of solutions for the $K$-XORSAT ensemble. Below $\alpha_d$ the space looks like a big connected 
 cluster whereas in the region $\alpha \in (\alpha_d, \alpha_s)$ the solution space breaks into exponentially many clusters far 
 away from each other.}\label{fig:space}
\end{figure}

\section{The Coupled $K$-XORSAT Ensemble} \label{cou}
This ensemble represents a chain of coupled underlying ensembles.
 Figure \ref{draw} is a visual aid but gives only a partial view.
 We consider $L-w+1$ clause positions $z \in \{0,1, \cdots, L-w\}$ and $L$ variable positions $z \in \{0,1, \cdots, L-1\}$.
At each variable position $z$, we lay down $n$ Boolean variables. Also, for each check position $z$, we lay down 
 $m= \lfloor \alpha n \rfloor$ clauses of length $K$. So in total we have $nL$ variables and $m(L-w+1)$ clauses.
Let us now specify how the set of edges, $E$, is chosen.
 Each clause $c$ at a position $z$,  chooses its $K$ variables via the following procedure. We first pick a position $z+k$ with $k$ uniformly random in 
the window $\{0,\cdots, w-1\}$, then we pick a variable  uniformly at random among all the variables located at position $z+k$, and finally 
we connect the clause and the variable. The value of $b_c$ is also chosen by flipping a fair coin.  
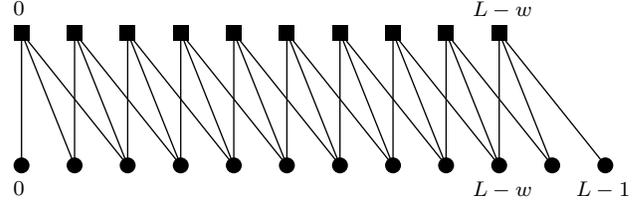
\begin{figure}[htp]
\begin{centering}
\input{draw}
\caption{{\small A representation of the geometry of the graphs with window size $w=3$ along
the ``longitudinal chain direction'' $z$. The ``transverse direction'' is viewed from the top. At each position there 
is a stack of $n$ variable nodes (circles) and a stack $m$ constraint nodes (squares). 
The depicted links between constraint and variable nodes represent stacks of edges.}}
\label{draw}
\end{centering}
\end{figure}
This ensemble is called the (spatially) coupled $K$-XORSAT ensemble and an instance of it is called a coupled formula.

It is also useful to consider another ensemble of coupled graphs where positions are placed on a ring. 
This ensemble is called the ring ensemble and is obtained as follows.
 We consider $L$ clause positions $z \in \{0,1, \cdots, L-1\}$ and $L$ variable positions $z \in \{0,1, \cdots, L-1\}$.
At each variable position $z$, we lay down $n$ Boolean variables. Also, for each check position $z$, we lay down 
 $m= \lfloor \alpha n \rfloor$ clauses of length $K$. So in total we have $nL$ variables and $mL$ clauses.
 Each clause $c$ at position $z$,  chooses its $K$ variables via the following procedure. We first pick a position $\mod(z+k,L)$ with $k$ uniformly random in 
the window $\{0,\cdots, w-1\}$, then we pick a variable node uniformly at random among all the variables located at position $z+k$, and finally 
we connect the clause and the variable. The value of $b_c$ is also chosen by flipping a fair coin.  
It can be easily seen that by picking a random ring formula and removing all of its clauses that are placed at positions $L-w+1, \cdots, L-1$ 
we generate a  coupled formula.


\subsection{Threshold Saturation}\label{thresh_sat}
The peeling algorithm can be used for the coupled and ring formulae in the same manner as explained above. 
We denote by $\alpha_{d, L, w}(K)$ and $\alpha_{d,L,w}^{\rm ring} (K)$ the threshold for the emergence w.h.p of a non-empty $2$-core
 for the coupled and ring ensembles. We also denote the SAT/UNSAT threshold for these ensembles by 
  $\alpha_{s, L, w}(K)$ and $\alpha_{s,L,w}^{\rm ring} (K)$, respectively.

Let us first consider the coupled ensemble. 
A similar message passing analysis as above yields  a set of one-dimensional coupled recursions 
\begin{equation}
 x_z^{t+1} = 1- \frac{1}{w}\sum_{l=0}^{w-1}
 \exp \bigl\{-\alpha K  ( \frac{1}{w} \sum_{k=0}^{w-1} x_{z+k-l}^t)^{K-1} \bigr\},
\label{pure-literal-rec-cou}
\end{equation}
with boundary values $x_z^t=0$ for $z \geq L$ and $z<0$.  
This recursion results in the one-dimensional fixed point equations
\begin{equation}
 x_z = 1- \frac{1}{w}\sum_{l=0}^{w-1}
 \exp \bigl\{-\alpha K  ( \frac{1}{w} \sum_{k=0}^{w-1} x_{z+k-l})^{K-1} \bigr\},
\label{pure-literal-fixed-point-cou}
\end{equation}
with boundary values $x_z^t=0$ for $z \geq L$ and $z<0$.  
We recall that  $\alpha_{d, L, w}(K)$ is the highest clause density for which the fixed 
point equation \eqref{pure-literal-fixed-point-cou} 
admits a unique solution that is the all-zero solution. 
\begin{lemma}
We have
\begin{equation}\label{thr_cou}
\lim_{w \to \infty} \lim_{L \to \infty}  \alpha_{d,L,w}(K)=\lim_{L \to \infty}  \alpha_{s,L,w}(K)=\alpha_s(K). 
\end{equation}
\end{lemma}
\begin{proof}
The fact that $\alpha_{s,L,w}$ tends to $\alpha_s$ as $L$ grows large, follows from the interpolation arguments of \cite{csp}. For the other limit, from \eqref{alpha_s} it can be shown that $\alpha_s$  corresponds to the potential threshold (defined in \cite{henry}) of the scalar recursion 
\eqref{pl_rec}. Hence, it follows from \cite[Theorem 1]{henry} that  $ \lim_{w \to \infty} \lim_{L \to \infty} \alpha_{d,L,w}$ tends to $\alpha_s$.
\end{proof}
As a result, as $L$ and $w$ grow large the peeling algorithm succeeds at densities very close to $\alpha_s(K)$.
Table~\ref{table-xor} contains some numerical predictions of $\alpha_{{\rm d},L,w}(K)$. 
\begin{table}
\centering
\begin{tabular}{c c c c c c c c c  }
 K & $3$ & & $4$  & & $5$ & & $7$   \\
\hline
$\alpha_s$           &  $0.917$    & &  $0.976$    & &   $0.992$       &&  $0.999$     \\
\hline
$\alpha_{d}(K)$ & $0.818$ & & $0.772$ & & $0.701$ & & $0.595$  \\
\hline
$\alpha_{d, L=80, w=5}(K)$   &  $0.917$ & & $0.977$ & &  $0.992$ & &$0.999$  \\
\end{tabular}
\caption{{\small {\it First line}: phase transition threshold for $K$-XORSAT. 
{\it Second line}: peeling threshold for the  uncoupled ensemble. {\it Third line}:  peeling threshold for the coupled ensemble with $w=5$, $L=80$.}}
\label{table-xor}
\end{table} 
For the ring ensemble, the fixed point equation for the peeling algorithm become
\begin{equation}
 x_z = 1- \frac{1}{w}\sum_{l=0}^{w-1}
 \exp \bigl\{-\alpha K  ( \frac{1}{w} \sum_{k=0}^{w-1} x_{\!\!\!\!\!\!\mod(z+k-l,L)})^{K-1} \bigr\}.
\end{equation}
It is easy to see that for $\alpha>\alpha_d(K)$, the above set of  fixed point equations admit a nontrivial solution in the following form.
For $z \in \{0,1, \cdots, L-1\}$, we have $x_z=x$, where $x$ is the largest solution the  FP equation in \eqref{pl_fp}. For $\alpha<\alpha_d$, it is also clear that 
there is only one solution which is the all-zero solution. 
Hence, for the ring ensemble we obtain for any choice of $L$ and $w$ 
\begin{equation} \label{thr_ring}
\alpha_{d,L,w}^{\rm ring}(K)= \alpha_d(K).
\end{equation}   
By combining \eqref{thr_ring} and \eqref{thr_cou}, one observes the following remarkable phenomenon. Let $L$ and $w$ be large but finite numbers such that $L \gg w$. 
For these choices of $L,w$ we have from \eqref{thr_cou} that $\alpha_{d,L,w}(K) \approx \alpha_s(K)$. Also, let $\alpha \in [\alpha_d, \alpha_{d,L,w}]$ and pick a formula from the ring ensemble. We deduce from \eqref{thr_ring} that such a formula has a non-trivial $2$-core. Furthermore, it can be shown that the $2$-core has a circular structure and for each position $z \in \{0, \cdots, w-1\}$,  it has $nV(\alpha,K)(1+o(1))$ variables $nC(\alpha,K)(1+o(1))$ clauses. Now, assume that from this $2$-core we remove all the clauses at positions $ L-w+1, \cdots,L-1$ (i.e., we open the ring, see Figure~\ref{fig:cores}) and run the peeling algorithm on the remaining graph. From \eqref{thr_cou} we deduce that the peeling algorithm succeeds on the remaining graph in the sense that it continues all the way until it reaches the empty graph.  Note here that the ratio of the clauses that we remove from the $2$-core is $\frac{w}{L}$ which vanishes as we choose $L \gg w$.         
  \begin{figure}[ht!]
 \begin{center} 
  \includegraphics[width=8cm]{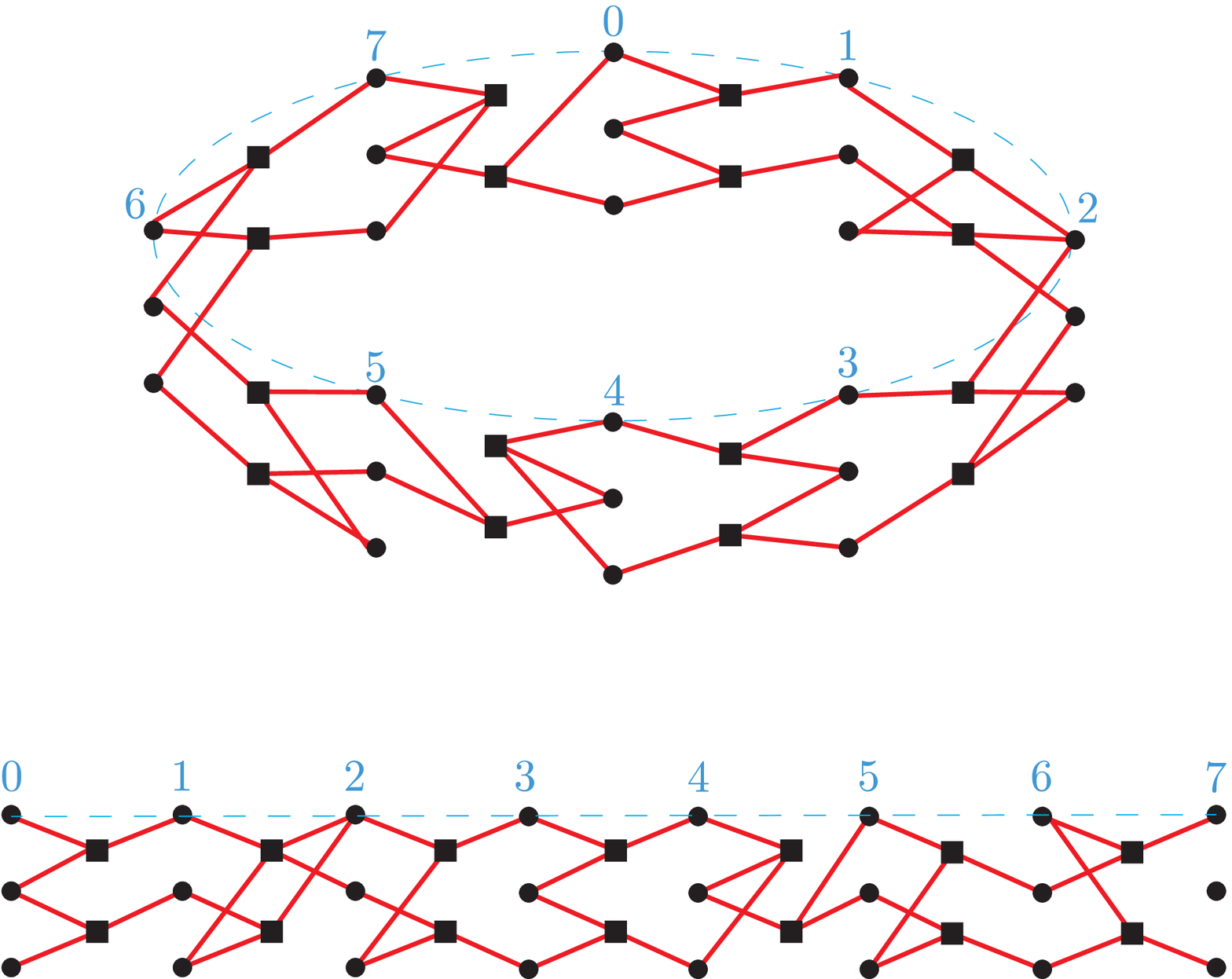}
 \end{center}
 \caption{For $\alpha \in (\alpha_d,\alpha_{d,L,w})$, a random formula from the ring ensemble has w.h.p a non-trivial $2$-core. The top figure is a simple example of a $2$-core associated to a ring formula with $L=8$ and $w=2$. When we open the $2$-core by removing the check nodes at positions $L-w+1, \cdots, L-1$ (the bottom figure), the remaining graph has w.h.p no $2$-core.  \label{fig:cores}}
\end{figure}

\subsection{The Set of Solutions}
We now focus on the geometrical properties of the space of  solutions of the coupled and ring formulas. Given the fact that for 
$\alpha \in (\alpha_d, \alpha_{s})$ a ring formula has a core, we deduce that for this region of $\alpha$ the  set of solutions of a ring formula resembles the set of solutions of an uncoupled formula which was explained in Section~\ref{un_sol}. 
 In other words, the  space of solutions of a ring formula
shatters into exponentially many clusters (see Figure~\ref{fig:space}). Each cluster corresponds to a unique solution of the $2$-core. Also, each cluster is itself connected and the distance between any two different clusters is $\Theta(nL)$.  
Now, assume $L$ and $w$ are large but finite numbers such that $L \gg w$. 
For these choices of $L,w$ we have from \eqref{thr_cou} that $\alpha_{d,L,w}(K) \approx \alpha_s(K)$. 
Let $\alpha \in (\alpha_d, \alpha_{d,L,w})$ and pick a formula from the coupled ensemble.  Let us denote this formula by $F$ and its set of solutions by $S$.
This formula w.h.p does not have a core. Also, we keep in mind that a coupled formula can be obtained from a typical ring formula by removing the 
clauses at the last $w$ positions. We denote such a ring formula by $F^{\rm ring}$ and its set of solutions by $S^{\rm ring}$.  We know that $S^{\rm ring}$ shatters into exponentially many clusters.   It is easy to see that $ S^{\rm ring} \subseteq S $. As a result $S$ contains all the clusters of $S^{\rm ring}$. Given these facts, how does the space of the space $S$ look like? In particular how are the two spaces $S$ and $S^{\rm ring}$  related? We now show that the space $S$ is a connected cluster.
 \begin{theorem}\label{main}
 Let $\alpha \in ( \alpha_d, \alpha_{d,L,w} )$. Consider a random coupled $K$-XORSAT formula and let $S$ be its set of solutions. The set $S$ is a connected cluster in the following sense. There exists a $B=B(\alpha,K)<\infty$ such that for any   two solutions $\underline{x},\underline{x'} \in S$, there exists a sequence of solutions $\underline{x}=\underline{x}_0, \underline{x}_1, \cdots, \underline{x}_r=\underline{x'}$ such that $d(x_i, x_{i+1}) \leq (e^L \log n)^B$.    
 \end{theorem}     
 {\em Proof sketch:} 
 The proof of this theorem essentially mimics the proof of Theorem 2 in \cite{mon} except for the last part. For the sake of briefness, we only give an sketch of the proof.  The proof goes by showing that the set of solutions of the equation $\mathbb{H} \underline{x}=0$, i.e. the kernel of the matrix $\mathbb{H}$, has a {\em sparse basis}. In other words,  there exists vectors $y_1, y_2, \cdots, y_I$ that span the space ${\rm kernel}(\mathbb{H})$, and each of the 
 vectors $y_i$ has a low weight, i.e., $w(y_i) \leq  (e^L \log n)^B $ where $w(\cdot)$ denotes the Hamming weight. We call such a basis a sparse basis. It is easy to see that if  such a basis exists for the
  space of solutions, then the result of the theorem holds.  
  
  We now proceed by explicitly constructing such a basis. 
  We first show that if the matrix $\mathbb{H}$ has no core, then the peeling procedure provides us with a natural choice of a basis for  ${\rm kernel}(\mathbb{H})$. We then show that such a basis is indeed sparse. In this regard, we consider an slightly modified, but equivalent, version of the peeling algorithm called the {\em the synchronous peeling} algorithm.  
 Given an initial formula (graph) $G$,  this algorithm consists of $T(G)$ rounds $t=1,2, \cdots, T(G)$. The residual graph at the end of round $t$ is denoted by $J_t$. We also let $J_0=G$. We denote the set of clauses, variables and edges removed at round $t$ by $(C_t,V_t, E_t)$. Hence for $t \geq 1$ we have $J_{t-1}=J_{t} \cup (C_t, V_t, E_t)$. At each round $t$, the algorithm considers the graph $J_{t-1}$ and removes all the variable nodes that have degree $1$ or less  together with all the clauses (if any) connected to these variables.   It is easy to see that synchronous peeling is somehow a compressed version of the peeling algorithm mentioned in Section~\ref{peel}. Assuming that the initial graph $G$ has no core, the final $J_{T(G)}$ is empty. 
 
 To ease the analysis, let us re-order the  clauses and the variables in the following way. We start from the clauses in 
 $C_1$  and order theses clauses (in an arbitrary way) from $1$ to $|C_1|$. We then consider clauses in $C_2$ and order them  
 (in an arbitrary way) from $| C_1| +1$ to $|C_1|+ |C_2|$ and so on.  We do the same procedure for the variable nodes but with the following additional ordering. Within each set $V_t$, the ordering is chosen in such a way that nodes that have degree $0$ in $J_{t-1}$ appear with a smaller index than the ones that have degree $1$.   Now, with such a re-ordering of the nodes in the graph, the matrix $H$ has the following fine structure. For the sets $P \subseteq C$ and $Q \subseteq V$, we let $\mathbb{H}_{P,Q}$ be the sub-matrix of $\mathbb{H}$ that consists of elements of $\mathbb{H}$ whose rows are $c \in P $ and columns are $i \in Q$. The matrix $\mathbb{H}$ can be partitioned into $T(G) \times T(G)$ block matrices $\mathbb{H}_{C_s, V_t}$ where $1 \leq s,t \leq T(G)$ such that for $s>t$, $\mathbb{H}_{C_s, V_t}$   is the all-zero matrix and the diagonal blocks $\mathbb{H}_{C_t, V_t}$ have a staircase structure. Here, by a staircase structure we mean that the set of columns of   $\mathbb{H}_{C_t, V_t}$ can be partitioned into $|C_t|+1$ groups $\mathcal{C}_0, \cdots, \mathcal{C}_{|C_t|}$ such that the columns in $\mathcal{C}_0$ are all-zero and the  columns in $\mathcal{C}_i$ have only their $i$-th entry equal to $1$ and the rest are equal to $0$.   Given such a decomposition of $\mathbb{H}$, it is now easy to see how one can find a basis for its kernel. In fact, the matrix $\mathbb{H}$ has essentially an upper triangular structure. With this structure, one can apply the method of back substitution \cite[Lemma 3.4]{mon} to solve the equation $\mathbb{H} \underline{x}=0$ and find the kernel of $\mathbb{H}$. Here, for the sake of briefness we just mention the final result.   We partition $V$ into a disjoint union $V=U \cup W$ in a way that $\underline{x}_W$ will be our set of independent variables and $\underline{x}_U$ will be the set of dependent ones (i.e., $\underline{x}_U$ can be expressed in terms of $\underline{x}_W$). The partition is then constructed by letting $W= W_1 \cup W_2 \cdots \cup W_{T(G)}$ and $U= U_1 \cup U_2 \cdots \cup U_{T(G)}$. For each $t$, we construct $W_t$ by using the staircase structure of $\mathbb{H}_{C_t, V_t}$. We recall that the columns of $\mathbb{H}_{C_t, V_t}$ have the partition 
 $V_t=\mathcal{C}_0 \cup \mathcal{C}_1 \cdots \cup \mathcal{C}_{|C_t|}$. We then construct $W_t$ as   $W_t=\mathcal{C}_0 \cup \mathcal{C}'_1 \cdots \cup \mathcal{C}'_{|C_t|}$, where $ \mathcal{C}'_i$ is constructed from $ \mathcal{C}_i$ by removing  an arbitrary element from it ($ \mathcal{C}'_i$ is empty if $|\mathcal{C}_i|=1$ ).  In other words, among the variables in $ \mathcal{C}_i$ we choose one as the dependent variable and let the others be independent variables in $W_t$.  We then let $U_t = V_t \setminus W_t$.  With the sets  $W$ and $U$ explained as above, let us reorder the variable in $V$ as $U$ followed by $W$, i.e., we reorder the variables such that we can write
  $\underline{x}= (\underline{x}_U, \underline{x}_W)$.  One can show that the columns of the matrix
\begin{equation} \label{BBtilde}
 \mathbb{K}=
\left[  \begin{array}{c}
\mathbb{H}_{C,U}^{-1} \mathbb{H}_{C,W} \\ \mathbb{I}
\end{array}
\right],
\end{equation}
form a basis for the set of solutions. Here,  the matrix $\mathbb{I}$ denotes the identity matrix of size $|W|= |V| - |C|$. Also, if $\mathbb{K}_{i,j}=1$ then we have 
$d_G(i,j) \leq T(G)$, where by $d_G(i,j)$ we mean the distance between variables $i,j$ in the graph $G$.  

It is now easy to show that the Hamming weight of any column of $\mathbb{K}$ is bounded above by the value $\max_{i \in V} | B_G(i, T(G)) | $, where by
$B_G(i, T(G))$ we mean the set of variables $j$ such that $d_G(i,j) \leq T(G)$.    In the last step, we argue that with high probability 
\begin{equation} \label{T(G)}
T(G) \leq v L + B_1 \log \log n,
\end{equation} 
where $v$ and $B_1$ are finite constants.
From \eqref{T(G)}, \cite[Lemma 3.11]{mon}, and the fact the  coupled ensemble has the same local structure as the un-coupled ensemble, we then deduce that
w.h.p  $\max_{i \in V} | B_G(i, T(G)) |  \leq e^{B_2 T(G))} \leq (e^L \log(n))^B $, where $B$ and $B_2$ are finite constants.  It remains to justify \eqref{T(G)}.  Consider the 
DE equations \eqref{pure-literal-rec-cou} starting from the initial point $x_z^0=1$ for $1\leq z \leq L-1$ and $x_z^0=0$ for $z$ at the boundaries. Let $\delta$ be a (very) small constant. It can be shown from \cite{KRU3} that  that there exists a constant $v=v(\alpha,K, \delta)<\infty$ such that $x_z^{v L} \leq \delta$ for  all $z \in \{0,1, \cdots, L-1\} $.  
 In other words, the effect of the boundary (i.e., $x_z^0=0$ for $z\geq L$ and $z <0$) propagates towards the positions at the middle in wave-like manner and with a speed $v$ and hence at time $t=vL$ all the values $x_z^t$ are small. Once the value of  $x_z^t$ is 
 sufficiently small then it converges to $0$ doubly exponentially fast. Hence, intuitively, the synchronous peeling algorithm needs w.h.p an extra $B_1 \log \log n$ steps to clear out the whole formula and the total time taken by peeling will be $vL+B_1 \log  \log n$ . Of course, this is just an intuitive argument. A formal analysis can be followed similar to  \cite[Lemma 3.11]{mon}.  
\subsection{An Intuitive Picture of the Sparse Basis}
As we mentioned in the end of Section~\ref{thresh_sat}, a ring formula with density $\alpha \in (\alpha_d, \alpha_{d,L,w})$ has a core. The core has a circular structure with roughly $nC(\alpha,K) $ clauses and $nV(\alpha,K) $ variables in each position $z \in \{0,1,\cdots, L-1\}$. Further, each two solutions of the core are different in $\Theta(nL)$ positions. Now, consider the formula obtained by removing the clauses at the last $w-1$ positions of the core (i.e., positions $L-w+1, \cdots,L-1$). We call such a formula the {\em opened core}. We know that the peeling algorithm succeeds on the opened core and from Theorem~\ref{main} its solution space is a connected cluster and  admits a sparse basis. So the distant solutions of the original core are now connected to each other via the new solutions spanned by this sparse basis. Our objective is now to see, at the intuitive level, how its spare basis looks like. 

All the variables in the opened core have degree at least two except the ones at the two boundaries (we call the first $w-1$  positions and the last $w-1$ positions the boundaries of the chain). Once the synchronous peeling algorithm begins, the effect of these low degree variables at the boundaries starts to propagate like a wave towards the middle of the chain. The algorithm evacuates the positions one-by-one with a constant speed $v$ approaching the middle \cite{KRU3}.   A simple, albeit not very accurate, analogy is a chain of properly placed domino pieces. Once we topple a boundary piece the whole chain is toppled with roughly a constant speed. 

Consider the peeling algorithm explained in Section~\ref{peel}. This algorithm removes the variables in the graph one-by-one. Each variable that is removed in this algorithm has either degree $0$ or $1$.  A variable that, at the time of being peeled,  has degree $0$ is called an \emph{independent} variable. A variable of degree $1$ is called a \emph{dependent} variable. One can easily see that the  definition of an independent (dependent) variable is equivalent to the definition given in the proof of Theorem~\ref{main}. In Theorem~\ref{main} we proved that the opened core has a sparse basis.   The number of elements of the basis is equal to the number of independent variables explored during the peeling algorithm.  Furthermore, there is a one-to-one correspondence between the independent variables and the elements of the sparse basis, as we explain now.  

Consider the synchronous peeling procedure defined in the proof of Theorem~\ref{main}.   The synchronous peeling procedure is a compressed version of the peeling algorithm in the following sense. At any step of  synchronous peeling, we peel  all the variables  in the remaining graph that have degree $0$ or $1$.  
Let us now denote the graph of the opened core by $G^*=(C^*,V^*,E^*)$.
Consider an independent variable $i \in V^*$ and assume that the variable $i$ is removed at step $t_i$ of the synchronous peeling algorithm. Let $\mathcal{H}_{G^*}(i,t_i)$ be the set of all the variables $u$ such that\footnote{We denote by $d_{G^*}(i,u)$ the graph distance between the variables $u$ and $i$ in the opened core.} $d_{G^*}(i,u) \leq t_i$ and $u$ is peeled at some time \emph{before} $i$. We also include in $\mathcal{H}_{G^*}(i,t_i)$  any check node (together with its edges) whose variables are all inside $\mathcal{H}_{G^*}(i,t_i)$.  Intuitively, $\mathcal{H}_{G^*}(i,t_i)$ corresponds to the history of the variable $i$ with respect to the peeling procedure. Figure~\ref{fig:treepeel}  illustrates these concepts via a simple expample. 
  \begin{figure}[ht!]
 \begin{center} 
  \includegraphics[width=8cm]{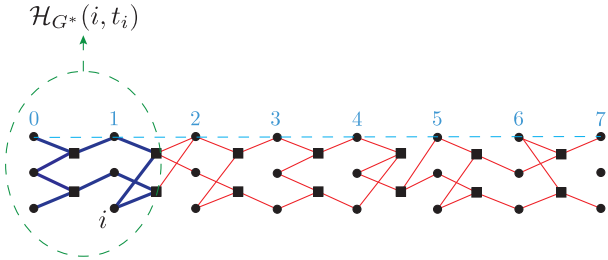}
 \end{center}
 \caption{Variable $i$ is an independent variable that is peeled off at the third step of the synchronous peeling algorithm, i.e., $t_i=3$.  The sub-graph  $ \mathcal{H}_{G^*}(i,t_i) $ consists of all the variables and checks of the opened core (together with the edges between them) that are peeled at some time before $i$ and whose distance from $i$ is less than $t_i$. }\label{fig:treepeel}
\end{figure}
As we explained above, the (synchronous) peeling procedure on the opened core  propagates like a wave from the boundaries towards the middle of the core, with a constant speed $v$. As a result, if the variable $i$ is at a (variable) position $p \in \{0,,1\cdots,L-1\}$, then we have\footnote{Of course, this relation is true for most of the independent variables in the opened core and there is a (vanishing) fraction of variables for which it takes $O(\log \log n)$ steps for them to be peeled off.}
$t_i \approx v p$. As a result, when $n$ is large and $n \gg L$, then $\mathcal{H}_{G^*}(i,t_i)$ is w.h.p a tree whose leaf nodes are located at one of the boundaries of the opened core (see Figure~\ref{fig:treepeel}).  Let us now see how the basis vector corresponding to the independent variable $i$ looks like. One can think of $\mathcal{H}_{G^*}(i,t_i)$ as a sub-graph or a sub-formula of $G^*$. Also, since we are solving the equation $H \underline{x} = \underline{0}$, a solution of  $\mathcal{H}_{G^*}(i,t_i)$ can naturally be extended (lifted) to a solution of $G^*$ by simply assigning $0$ to the variables in $G^*  \setminus \mathcal{H}_{G^*}(i,t_i) $.  
Consider a solution of $\mathcal{H}_{G^*}(i,t_i)$ for which the value that the variable  $i$ takes is $1$. Since the peeling succeeds on $\mathcal{H}_{G^*}(i,t_i)$ and $i$ is an independent variable, such a solution exists (one can find such a solution by assigning $1$ to $i$ and then backtracking on $\mathcal{H}_{G^*}(i,t_i)$). Such a solution, when extended to a solution of $G^*$  is the corresponding basis element for the variable $i$.      


\end{document}

%% file: definition.tex
\newtheorem{theorem}{Theorem}

\newcommand{\btheo}{\begin{theorem}}
\newcommand{\etheo}{\end{theorem}}
\newcommand{\bproof}{\begin{proof}}
\newcommand{\eproof}{\end{proof}}
\newtheorem{definition}[theorem]{Definition}
\newcommand{\bdefi}{\begin{definition}}
\newcommand{\edefi}{\end{definition}}
\newtheorem{fact}[theorem]{Fact}
\newcommand{\bprop}{\begin{fact}}
\newcommand{\eprop}{\end{fact}}
\newtheorem{corollary}[theorem]{Corollary}
\newcommand{\bcor}{\begin{corollary}}
\newcommand{\ecor}{\end{corollary}}
\newtheorem{example}[theorem]{Example}
\newcommand{\bex}{\begin{example}}
\newcommand{\eex}{\end{example}}
\newtheorem{lemma}[theorem]{Lemma}
\newcommand{\blemma}{\begin{lemma}}
\newcommand{\elemma}{\end{lemma}}
\newtheorem{remark}[theorem]{Remark}
\newcommand{\bremark}{\begin{remark}}
\newcommand{\eremark}{\end{remark}}
\newtheorem{conj}[theorem]{Conjecture}
\newcommand{\bconj}{\begin{conj}}
\newcommand{\econj}{\end{conj}}

\def\0{{\tt 0}} 
\def\1{{\tt 1}} 
\def\?{{\tt *}} 

%% file: draw.tex
\setlength{\unitlength}{1bp}%
\begin{picture}(200,80)(0,0)

\put(-20,0)
{
\put(-12,0){\rotatebox{0}{\includegraphics[scale=1]{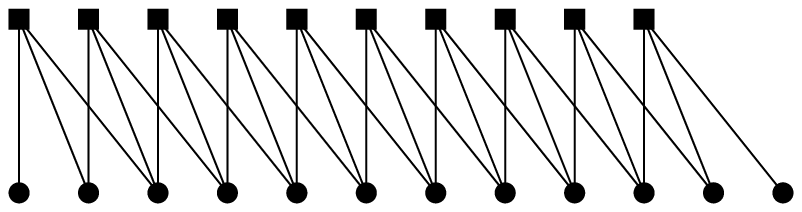}}}

    \footnotesize	
     \put(2,75){\makebox(0,0)[t]{$0$}}
       \put(173,75){\makebox(0,0)[tl]{$L-w$}}
    
      \put(2,7){\makebox(0,0)[t]{$0$}}
       \put(173,7){\makebox(0,0)[tl]{$L-w$}}
         \put(212,7){\makebox(0,0)[tl]{$L-1$}}

}

\end{picture}